\RequirePackage{fix-cm}
\documentclass{svjour3_1qbit}                    
\usepackage{amsmath}
\usepackage{amsfonts}
\usepackage{amssymb}
\usepackage{mathtools}
\usepackage{graphicx}
\usepackage{epstopdf}
\usepackage{pgf}
\usepackage{stmaryrd}
\usepackage[verbose]{wrapfig}
\usepackage{subfigure}

\usepackage{xcolor}
\usepackage{url}
\usepackage{authblk}

\usepackage[]{algorithm2e}
\usepackage{courier}
\usepackage{color}
\DeclarePairedDelimiter{\ceil}{\lceil}{\rceil}

\smartqed  

\begin{document}

\title{Systematic and Deterministic Graph Minor Embedding for Cartesian Products of Graphs}
\titlerunning{Systematic and Deterministic Embedding for Cartesian Products of Graphs}

\author{Arman Zaribafiyan        \and
        Dominic J. J. Marchand \and 
        Seyed Saeed Changiz Rezaei
}

\institute{ 1QB Information Technologies (1QBit) \\
              458-550 Burrard Street, Vancouver, BC, V6C 2B5, Canada\\
              Tel.: +1.646.820.8865\\ \\
              \email{arman.zaribafiyan@1qbit.com} \\
              \email{dominic.marchand@1qbit.com}\\
              \email{seyed.rezaei@1qbit.com} 
              }

\date{}

\maketitle

\begin{abstract}
The limited connectivity of current and next-generation quantum annealers motivates the need for efficient graph minor embedding methods. These methods allow non-native problems to be adapted to the target annealer's architecture. The overhead of the widely used heuristic techniques is quickly proving to be a significant bottleneck for solving real-world applications. To alleviate this difficulty, we propose a systematic and deterministic embedding method, exploiting the structures of both the specific problem and the quantum annealer. We focus on the specific case of the Cartesian product of two complete graphs, a regular structure that occurs in many problems. We decompose the embedding problem by first embedding one of the factors of the Cartesian product in a repeatable pattern. The resulting simplified problem consists in the placement and connecting together of these copies to reach a valid solution. Aside from the obvious advantage of a systematic and deterministic approach with respect to speed and efficiency, the embeddings produced are easily scaled for larger processors and show desirable properties for the number of qubits used and the chain length distribution. We conclude by briefly addressing the problem of circumventing inoperable qubits by presenting possible extensions of our method.
\keywords{graph minor embedding, Cartesian product, quantum annealing}
\end{abstract}

\section{Introduction}

The majority of the interesting combinatorial optimization problems are hard to solve. Graph similarity, graph partitioning, graph colouring, resource allocation, and scheduling problems are among those combinatorial optimization problems proven to be NP-hard \cite{Partitioning,Coloring,1QBit_JSP}. Many of these problems have significant real-world applications which make them especially interesting. For example, determining similarities between graphs is a challenging problem that occurs when comparing the structures of different molecules and is thus of great importance in drug design applications \cite{1QBit_GS}. Many of these problems can be formulated as quadratic unconstrained binary optimization (QUBO) problems, which can be solved on specialized quadratic solvers. 

One  type of quadratic solver which has garnered considerable attention in the past few years is the quantum annealing processor manufactured by D-Wave Systems \cite{DW_manufactured_Spin,Scalable_Arch}. In essence, the processor is a specialized quantum device that samples low-energy configurations of a set of Ising spin variables \mbox{$\mathbf s \in \{-1, +1\}^n$} that serve as quantum registers or qubits. It is designed to solve a binary input problem formulated as an Ising Hamiltonian specified by a pair $(\mathbf h, \mathbf J)$, where $\mathbf h\in \mathbb R^n $ is a vector of local fields acting on the spin variables, and $\mathbf J \in \mathbb R ^ {n\times n}$ is a symmetric matrix of quadratic couplings between these spins. The objective function to be minimized is specified by the energy $E(\mathbf s)$ of the spin configuration $\mathbf s$ and given by 
\begin{equation}\label{eq_Es}
E(\mathbf s) = E(\mathbf s, \mathbf h, \mathbf J) = \mathbf s^T \mathbf J\mathbf s + \mathbf s^T  \mathbf h.
\end{equation}
We note that an Ising problem can be formulated as a QUBO problem through a simple linear transformation. Therefore, the quantum annealer is equivalently considered to be a quadratic unconstrained binary optimizer that minimizes a quadratic objective function $\mathbf Q$ given by 
\begin{equation}\label{eq_Ex}
E(\mathbf x) = E(\mathbf x,  \mathbf Q) = \mathbf x^T \mathbf Q \,\mathbf x
\end{equation}
over the discrete configuration space of a set of qubits $\mathbf x \in \{0, 1\}^n$. The solver has limited connectivity between its qubits such that not all pairs can be coupled together. In other words, only a subset of the terms of $\mathbf J$ or $\mathbf Q$ are allowed to be non-zero. For this reason, the structure of the problem to be solved must be mapped to the architecture of the solver, a process called \emph{embedding} \cite{choi2011_minor_emb}. 

In this paper, we will treat both the input problem and the solver as graphs. An input problem of interest, formulated as either a QUBO or Ising problem, can be represented as a graph $G=(V, E)$, where $V$ is a set of vertices representing either the logical variables or physical qubits, and $E$ is a set of edges representing the interactions between them. For the case of an Ising problem, the vertices of this graph are the variables $s_1,\dots, s_n$, while the set of edges is created by adding one edge for each pair of vertices $s_i$ and $s_j$ for which $J_{ij}$ is non-zero. On the other hand, the processor's architecture is described by the hardware graph $\mathcal{C}$. This graph represents the available physical qubits or registers and shows how they are coupled together on the processor. The earlier D-Wave Two hardware graph (see Figure \ref{fig:Chip509}) has 512 physical qubits and each qubit is adjacent to at most 6 others. D-Wave's nearly regular hardware graph, the \emph{Chimera} graph, is denoted by $\mathcal{C}_{N,M,L}$ and constructed as an $N \times M$ grid of $K_{L,L}$ bipartite blocks, as defined in \cite{boothby2015_fast}.

In order to embed the desired Ising model into the processor, the graph $G$ should be a subgraph of graph $\mathcal{C}$. A mapping of the input graph to the target graph is called a \emph{direct embedding}. Seeking a direct embedding places stringent constraints on the size and connectivity of the input graph. Alternatively, we can seek a \emph{graph minor embedding}, which is a specific type of mapping where we further allow adjacent vertices of the target graph to be contracted into larger effective vertices, called chains. In this more general case, the graph $G$ should be a subgraph of a graph minor of $\mathcal{C}$. For a detailed description of graph minor embedding, see \cite{choi2011_minor_emb}. In simple terms, a chain is created through the  addition of strong penalty terms to the objective function such that the variables involved are forced to take the same value. In the Ising formulation, this is achieved by applying a strong ferromagnetic coupling between any two adjacent vertices $i$ and $j$ of $G$ in the same chain.

In the most general case, where no assumption is made about the input and target graphs, seeking a graph minor embedding is an NP-hard problem \cite{cai_2014_practical}. Practically, this means that as the size of the graphs increases, the problem of finding a valid embedding quickly becomes prohibitively computationally expensive. To avoid having to solve an NP-hard problem with each embedding, we could use the fact that the structure of the solver is usually known in advance (here, it is a Chimera graph structure possibly less a few inoperable qubits and couplers). This means that polynomial solutions to the embedding problem remain achievable. While such an exact method exists, its poor scaling still renders it unusable for graphs larger than 10 vertices \cite{Adler2011}. As a result, the most widely used embedding algorithms, such as the one  introduced by Cai, Macready, and Roy \cite{cai_2014_practical}, are heuristic in nature, compromising on embedding quality in order to achieve polynomial running times with a more favourable scaling. Even then, finding an embedding is typically very time consuming. This is further exacerbated by other limitations of the analogue quantum device which can lead to highly variable performance, depending on the quality of the embedding, prompting the need to run the heuristic multiple times in order to select the best solution. Although less than ideal, heuristic solutions have proven to be mostly satisfactory for the exploratory work conducted on previous-generation quantum annealers, provided that sufficient computation time could be allocated for embedding. With the recent introduction of a 1000-qubit annealer, however, we are quickly reaching the point where more-scalable solutions are needed. In fact, quantum annealing making the leap from a nascent technology of purely academic interest to a useful mainstream tool is conditional on the availability of fast embedding methods that will not eradicate any potential quantum speed-up. 

The most promising next-generation embedding methods should be systematic and scalable. It is unlikely that such properties will be attained for truly general approaches, and advances will come by exploiting not only the structure of the target architecture, but also the structure of specific problems.  We believe that so long as embedding is needed, although general approaches are useful, it is with application-specific and systematic graph embedding approaches that the full potential of quantum hardware will be realized. The path to better or faster embedding algorithms, therefore, lies in restricting the graph minor embedding problem to specific cases. The triangular embedding of complete graphs  \cite{choi2011_minor_emb}, later generalized by Boothby, King, and Roy's approach of fast clique embedding for complete graphs \cite{boothby2015_fast}, epitomizes this application-specific approach and creates a systematic embedding for fully connected problems on the Chimera graph architecture. The embeddings produced have equal-length chains and are general because any graph is a subgraph of a complete graph. Unfortunately, this approach is wasteful for applications that do not require a fully connected graph, limiting the size of problems embeddable with this method. 

The first step in devising new embedding methods is to identify a common structure across many problems that can be exploited advantageously. As we will show below, a recurring graph structure which appears in the quadratic formulation of many of the NP-hard optimization problems mentioned above is the Cartesian product of two graphs. Cartesian products, being both ubiquitous and highly structured, are attractive targets for the type of improved methods we are advocating. One of the main contributions of this research is the analytical identification of this regularity and structure in the QUBO problem formulation of important families of NP-hard optimization problems. One of the most important advantages of this contribution is that it enables us to reuse the found embeddings for problems with similar structures. The ability to reuse embeddings reduces the computational complexity of the embedding process to a one-time cost per family of problems.

In this paper, we describe a procedure for embedding a Cartesian product of two graphs into a Chimera graph. The vertex set of the Cartesian product $G_1 \square G_2$ of two graphs $G_1 = (V_1,E_1)$ and $G_2=(V_2,E_2)$ is the Cartesian product of the vertex sets of the individual graphs. In the resulting graph, two vertices $(v_1, v_2)$ and $(u_1, u_2)$ are adjacent if and only if $v_1 = u_1$ and $v_2$ is adjacent to $u_2$ \textbf{or} $v_2 = u_2$ and $v_1$ is adjacent to $u_1$. Denoting the adjacency matrices of graphs $G_1$ and $G_2$ by $A_1$ and $A_2$, and having $n_1 := |V_1|$ and $n_2 := |V_2|$, we can compute the adjacency matrix of $G_1\Box G_2$, that is, $A_{G_1\Box G_2}$, in terms of the adjacency matrices of $G_1$ and $G_2$ as follows:
\begin{equation}
A_{G_1\Box G_2} =  I_{n_1}\otimes A_{G_2} + A_{G_1}\otimes I_{n_2}
\end{equation}

For the sake of generality with respect to embedding, for the remainder of this paper, we look into the Cartesian product of complete graphs.

\section{Identifying the Cartesian product of complete graphs (CPCG)}

The Cartesian product of graphs can appear in many application-driven problems. To exploit this structure, however, we first need to either infer its presence from the problem's QUBO form or preserve it as we formulate the problem from the outset. An alternative structure-preserving QUBO problem formulation can be found in \cite{hedayat}. Very efficient algorithms for identifying Cartesian products in arbitrary graphs have been proposed. For example, Imrich and Iztok proposed an exact algorithm \cite{Imrich2007_cp_linear} with linear scaling in terms of the number of edges for both the running time and memory requirement by using a clever edge-labelling technique. Nevertheless, it is useful to look at how Cartesian products occur when formulating optimization problems where doubly indexed binary variables are used. This is what we consider below. 

Suppose we have a QUBO problem where the variables are doubly indexed binary variables $x_{ik}$, where $1 \leq i \leq N$ and  $1 \leq k \leq K$. Such a structure occurs, for example, in the $K$-way graph partitioning problem. We formulate this graph partitioning problem as follows. Given a graph $G = (V, E)$ with $N$ vertices,  we want to divide the vertex set into $K$ partitions, where $K$ is a positive integer, such that the sum of edges inside the partitions is maximized. Let $A$ be the adjacency matrix of the graph $G$ built from the edge set $E$. For every vertex $i$ and partition $k$, the optimization variable $x_{ik}$ is $1$ if vertex $i$ is in partition $k$, and $0$ otherwise. Furthermore, without loss of generality, we assume that $N$ is divisible by $K$, and we let $P = N/K$. The objective is to find the assignment of vertices to partitions (i.e., a 0-1 configuration of $x_{ik}$'s) that maximizes the number of intra-partition edges and satisfies the following constraints:
\begin{enumerate}
\item \textbf{Orthogonality constraint}: each vertex must be assigned to one and only one partition
\item \textbf{Cardinality constraint}: each partition must have the same number of vertices assigned to it
\end{enumerate}
We note that for the case where $N$ is not divisible by $K$, the second constraint is relaxed such that the size of each partition should not differ from all others by more than one vertex. This problem can be formulated as the following optimization problem:
\begin{align}
\max\limits_{\textbf{x}=\{x_{i k}\}} \:\: & \sum\limits_{k=1}^{K}\left(\sum\limits_{(i_1,i_2) \in E} x_{i_1 k}x_{i_2 k}\right)  \nonumber \\
\text{s.t.:} \:\: &\sum\limits_{i=1}^N x_{ik} = P, \quad\quad \forall \: 1 \leq k \leq K \nonumber \\
&\sum\limits_{k=1}^K x_{ik} = 1, \quad\quad \forall \: 1 \leq i \leq N \nonumber \\
&x_{i,k} \in \{0,1\} , \quad\quad \forall \: i, k
\end{align}

\begin{figure}[hbtp]
\centering
\includegraphics[scale=1.0]{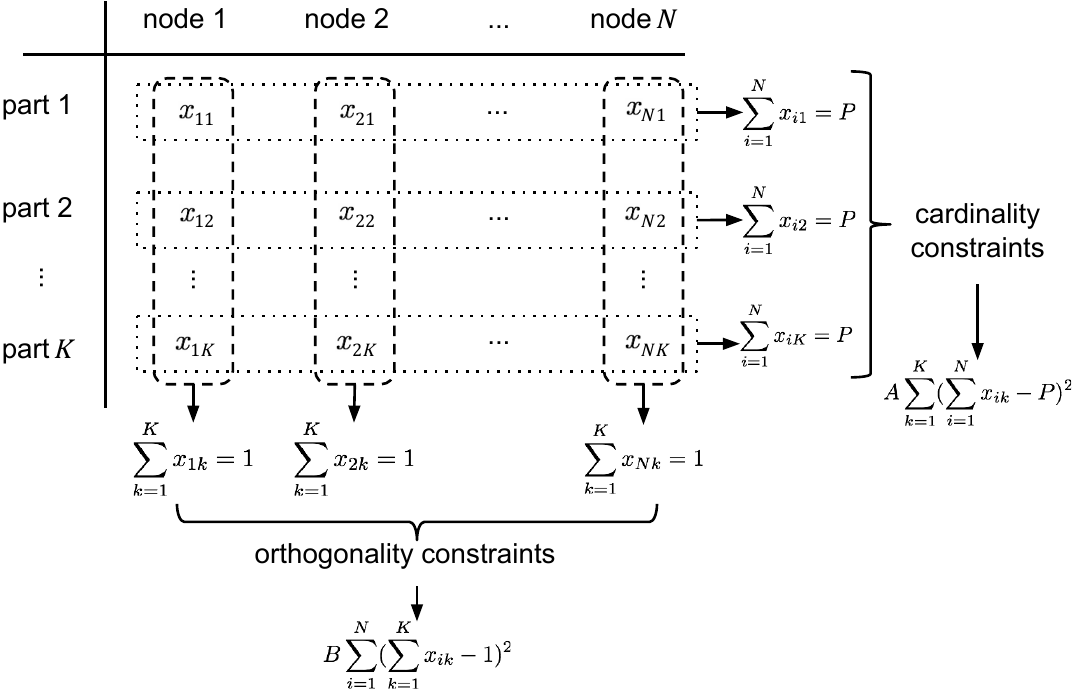}
\caption{A matrix representation of doubly indexed variables $x_{ik}$ in an example of a $K$-way partitioning problem on a graph with $N$ nodes.  The matrix representation illustrates how subsets of variables contribute to specific orthogonality and cardinality constraints.}
\label{fig:doubly_indexed}
\end{figure}

In order to formulate this problem as a QUBO problem appropriate for the annealer, we rewrite the objective function that needed to be maximized into an objective function to be minimized, and implement the equality constraints as quadratic penalty terms. The resulting QUBO problem is equivalent to the previous constrained optimization problem for appropriately chosen penalty constants $A$ and $B$:
\begin{align}
\min\limits_{\textbf{x}=\{x_{ik}\}} \Bigg[ \:\: - \sum\limits_{k=1}^{K}\sum\limits_{(i_1,i_2) \in E} x_{i_1k}x_{i_2 k} &+A\sum\limits_{k=1}^{K}\left(\sum\limits_{i=1}^N x_{ik} - P\right)^2  \\ \nonumber
&+B\sum\limits_{i=1}^{N}\left(\sum\limits_{k=1}^K x_{ik} - 1\right)^2 \Bigg]
\end{align}

The Cartesian product's structure is easily observed by constructing the QUBO problem graph for this partitioning problem. It can be built by reading the above QUBO objective function directly. The QUBO problem graph has a vertex for each doubly indexed binary variable, and an edge for each quadratic term of the objective function. The following quadratic terms are found:
\begin{enumerate}
\item For a fixed $k$, the first summation creates a quadratic term $x_{i_1k}x_{i_2 k}$ if $(i_1,i_2) \in E$.
\item For a fixed $k$, the second summation creates a quadratic term $x_{i_1 k}x_{i_2 k}$ for all $1 \leq i_1 < i_2 \leq N$ (after expanding the square of the sum).
\item For a fixed $i$, the third summation creates a quadratic term $x_{ik}x_{ik'}$ for all $1 \leq k < k' \leq K$ (after expanding the square of the sum).
\end{enumerate}
This correspondence between the quadratic terms and the edges in the QUBO problem graph results in the fact that any subset of vertices corresponding to a fixed $i$ or $k$ induces a complete graph on the problem graph. From this, we conclude that the resulting QUBO problem graph is a Cartesian product of two complete graphs $K_N \square K_K$. Figure \ref{fig:doubly_indexed} shows how grouping terms for a fixed partition $k$ and for a fixed vertex $i$ can assist in identifying the structure in the final QUBO formulation.

A similar argument can be used for any other input problem with doubly indexed variables to identify whether there exists a product graph structure in the resulting QUBO problem graph. In general, an input problem with doubly indexed variables where the objective function and constraints are defined on subsets of variables where one index is fixed will end up with QUBO problem graphs which are subgraphs of Cartesian products of complete graphs. Graph partitioning, graph colouring, and size-constrained clustering are important examples of such problems. In addition to these problems, Cartesian product structures have applications in error-correction for adiabatic quantum computation. Recent research has shown that using the Cartesian product of graphs as an error-correcting scheme reduces the time to solution for certain families of problems \cite{Lidar}.

\section{Description of CPCG Embedding}\label{sec:descr_CPCG}

We have mentioned that a systematic embedding relies in part on the regularity of the target graph's architecture. Our method is general and can be adapted to different architectures provided they can be described as a regular lattice of unit cells. Nevertheless, it will be convenient to restrict the presentation of our method to a specific case. The Chimera hardware graph is the obvious choice, as it describes the architecture of the only commercially available quantum annealer. The D-Wave Two processor uses a 512-qubit Chimera graph $\mathcal C_{8,8,4}$, and the newer D-Wave 2X uses a 1152-qubit Chimera graph $\mathcal C_{12,12,4}$.

We consider the Cartesian product of two complete graphs with sizes $m$ and $n$, that is, $K_m\Box K_n$, as the input graph. It is noteworthy that $K_m\Box K_n$ has $n$ distinct copies of $K_m$ as well as $m$ distinct copies of $K_n$ as induced subgraphs. 
We propose to first embed one copy of either $K_m$ or $K_n$, say $K_m$, into a repeatable unit which we call a \emph{nexus}. More precisely, a nexus consists of a collection of adjacent unit cells of the Chimera graph. The regularity of the grid architecture then allows for the embedding of $n$ copies of $K_m$ by simply placing one \emph{nexus instance}  on the grid for each of them. We are left with the problem of choosing the exact placement of these instances and connecting them together to realize the full Cartesian product. We call these inter-nexus connections \emph{buses} and their arrangement \emph{the bus configuration}. We have thus not only chosen a simpler high-level description of the original embedding problem, but also implicitly decomposed the problem into two subproblems: the nexus selection, and the nexus instances and bus configuration. We will now look at these subproblems in more detail and describe how they can be implemented to also achieve a scalable embedding strategy with advantageous properties.

\begin{figure}
\centering     
\subfigure[A possible nexus choice for a $K_8$ graph]{\includegraphics[scale=.273]{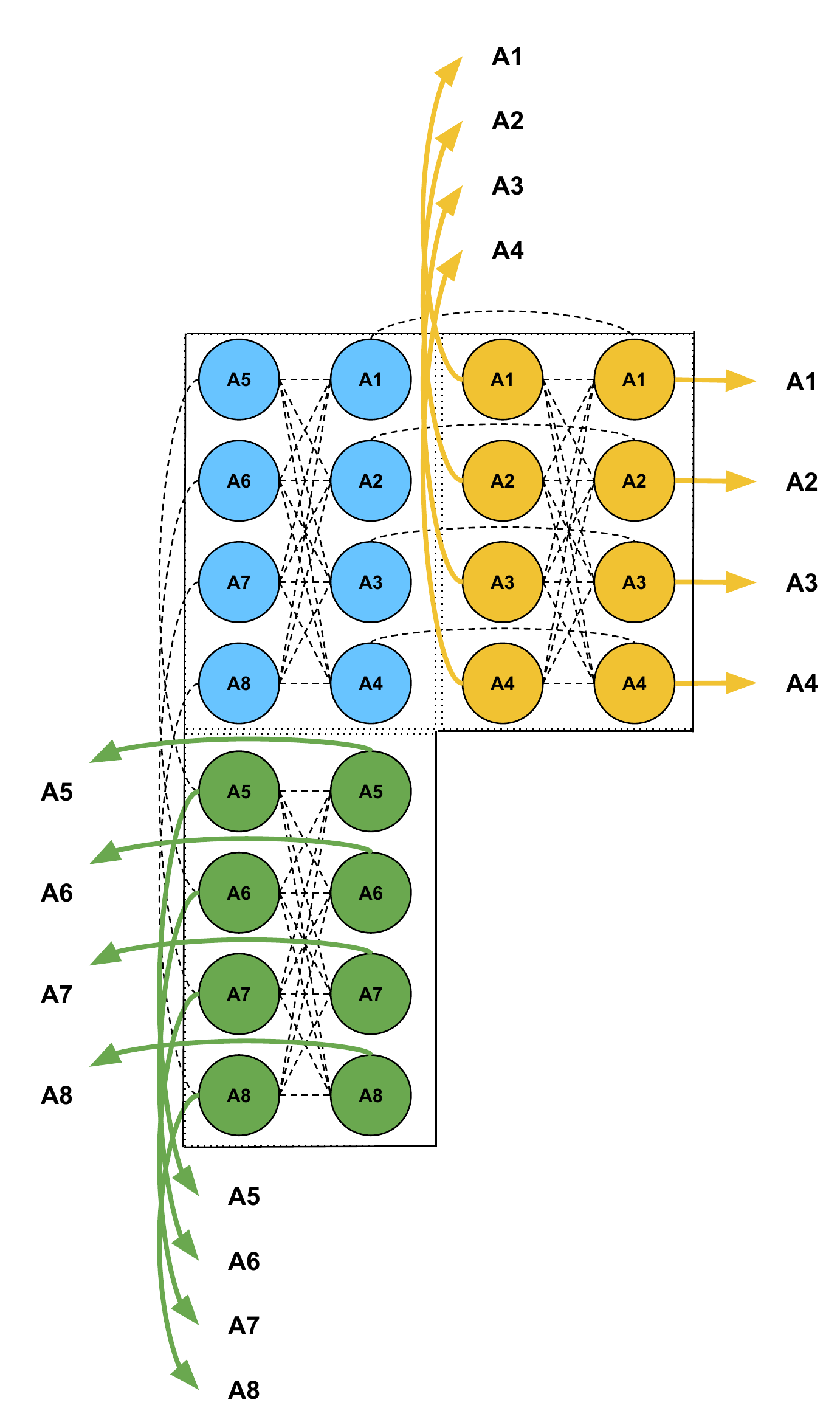} }
\subfigure[An embedding for $K_8 \square K_7$]{\includegraphics[scale=.55]{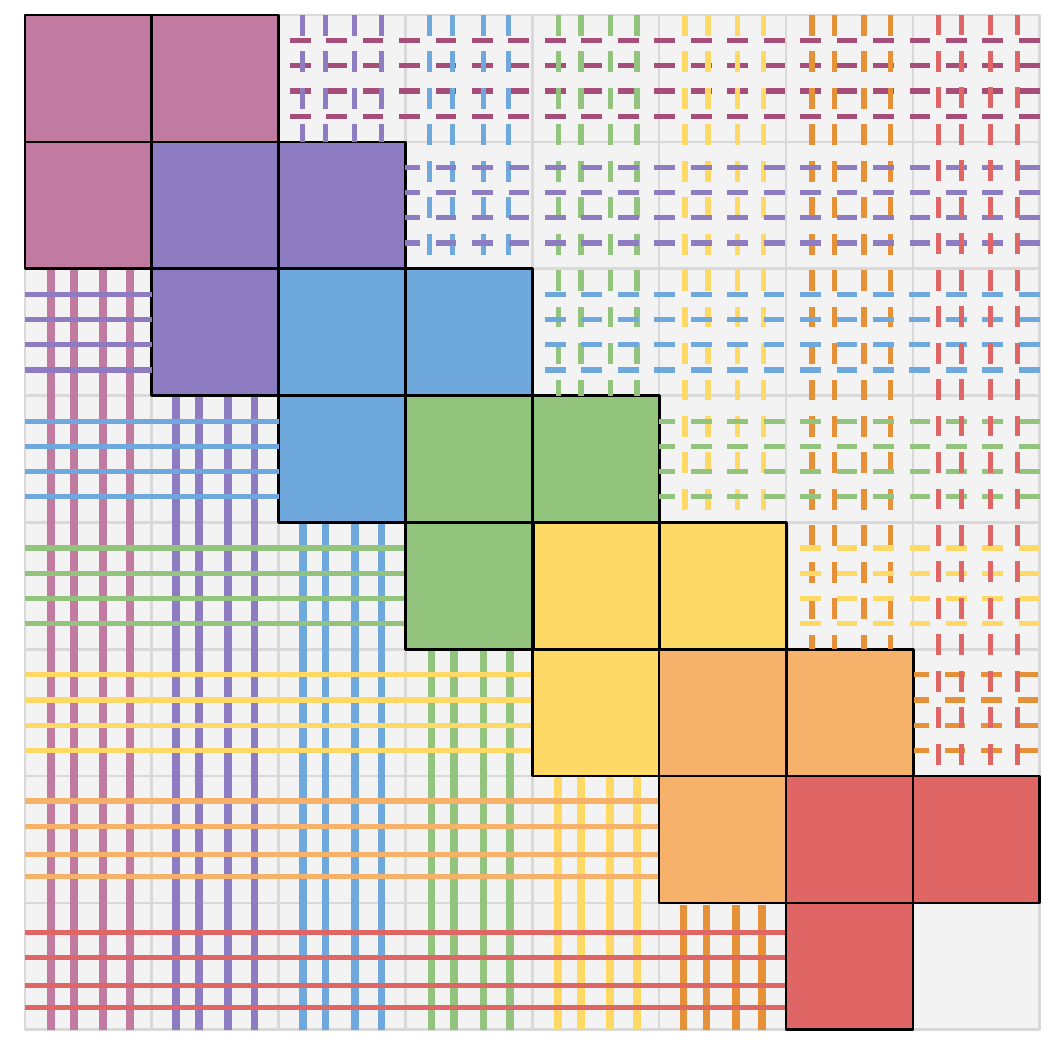} }
\caption{(a) A possible nexus choice for a $K_8$ graph on an ideal Chimera chip with $L=4$ using the triangular embedding method. (b) A valid CPCG embedding for the product $K_8 \square K_7$ using the nexus shown in (a). The embedding is shown for a Chimera target graph $\mathcal{C}_{8,8,4}$.}
\label{fig:nexus_emb_and_bus_config}
\end{figure}

We use the specific case of embedding $K_8\Box K_n$ on $\mathcal C_{N,N,4}$ as a demonstration. Figure \ref{fig:nexus_emb_and_bus_config}a shows how three adjacent unit cells of the Chimera graph, along with their couplers, are used for embedding $K_8$, and constitutes our preferred choice of nexus for hosting $K_8$. This embedding is essentially the same as the default triangular embedding for $K_8$ \cite{Scalable_Arch}. Figure \ref{fig:nexus_emb_and_bus_config}b illustrates an embedding pattern based on this choice of nexus and how buses were placed to realize $K_8\Box K_7$ on $\mathcal C_{8,8,4}$. 

\subsection{Nexus selection}

The nexus shape depends on the structure of the graph to be embedded as well as on the Chimera graph's architecture. Embedding a nexus can be viewed as a much smaller embedding problem with some added constraints pertaining to providing the appropriate connections to all variables through dedicated \emph{bus interfaces}. The definition of these interfaces allows some level of encapsulation by abstracting away the details of the nexus embedding for the rest of the method. On the target architecture, an interface is a high-level object that stands for the couplers coming out of the nexus and an indication of which variables are attached to them. In the high-level description, on the other hand, it serves as an attachment point for a bus extending a specific set of variables. Redundant interfaces can be defined if needed. It is important to note that the definition of the nexus interfaces will determine the optimal nexus placement and bus configuration. We can therefore seek a nexus, and then list the interfaces available, or we can require a set of interfaces as a constraint. To demonstrate the method, we require only that all variables be accessible through an interface, leaving more-advanced considerations for future work.

Solving the limited problem of finding a valid nexus could potentially benefit from other known embedding techniques, including heuristics methods. For the case at hand, embedding a complete graph on as few Chimera blocks as possible, the triangular embedding algorithm proposed in \cite{Scalable_Arch} represents an attractive option. This choice exploits the Chimera graph's large automorphism group (see \cite{cai_2014_practical}) and its resulting high level of symmetry. Figure \ref{fig:nexus_emb_and_bus_config}a shows that for a Chimera graph with $L=4$, three adjacent unit cells are sufficient to build a nexus for $K_8$. The bipartite nature of the Chimera block naturally partitions the set of variables corresponding to the nodes of $K_8$ into two subsets, each having four redundant interfaces: two that face downward and to the left, and two that face upward and to the right. For example, we partition the corresponding vertex set $\{A_1,\dots,A_8\}$ of the $K_8$ nexus shown in Figure \ref{fig:nexus_emb_and_bus_config}a into two subsets $\{A_1,A_2,A_3,A_4\}$ and $\{A_5,A_6,A_7,A_8\}$.  

Each subset of variables is assigned a number of redundant interfaces available for building the inter-nexus connections. An interface is therefore a set of connection points, called \emph{terminals}, corresponding to the variables of the subset and placed on a specific face of the nexus. For example, in Figure \ref{fig:nexus_emb_and_bus_config}a, a vertical bus connects to an interface representing the subset of variables $\{A_5,A_6,A_7,A_8\}$ and extend them downward, and a horizontal bus connects to a second interface on the same subset and extends them leftward. Similarly, two buses extends the subset $\{A_1,A_2,A_3,A_4\}$ in the opposite directions.

\subsection{Nexus instance placement and bus configuration}

We have slightly simplified the embedding problem by introducing a high-level description involving nexuses and buses.  We now need to solve that high-level problem. Fortunately, the number of degrees of freedom has been greatly reduced compared to that of the original problem. A tailored search algorithm can be implemented based on tabu search or simulated annealing, for example. The allowed steps or updates in configuration spaces are easily derived from the target architecture and its symmetries. We leave such a general solution for future work, however, and focus on a systematic method that works very well for embedding complete graphs on the Chimera architecture, inspired in part by the rooks problem \cite{kaplansky1946}.

We call the area of the Chimera graph not occupied by nexus instances the \emph{bus space}. First introduced near the beginning of Section \ref{sec:descr_CPCG}, a \emph{bus}, more precisely, is a set of parallel paths leaving from a nexus interface, with one path per variable (see Figure \ref{fig:nexus_emb_and_bus_config}b). Each path is assigned to a specific variable. Two buses can be linked together at a bus \emph{junction}. Locating the nexus instances hosting multiple copies of the complete graph $K_m$ on the diagonal of the Chimera graph divides the bus space into disjoint bus spaces. A unit cell of the Chimera graph where two buses meet can be used as a junction. 

The Chimera graph's bipartite structure and the proposed triangular embedding naturally invite a partitioning of the variables of $K_m$ into two subsets. We therefore seek a placement of the nexus instances along a line that would divide the bus space into two bus subspaces, providing access to both subspaces to each nexus instance. The L-shape of the nexus also lends itself to a more efficient tiling if we place these instances along the diagonal of the Chimera graph. 

Next, we need to extend the nexus interfaces to build the connectivity required by the full Cartesian product. For each subspace, this implies connecting each nexus instance through one of its interfaces to an interface of each other nexus instance in the same subspace. A valid configuration inspired by the rooks problem is to attach both a vertical and a horizontal bus that run to the edge of the chip. This creates a rectilinear grid where each pair of nexus instances meet at a single unit cell. We use each of the created disjoint bus spaces to establish the required connections for copies of $K_n$ in $K_m\Box K_n$. This embedding of copies of $K_n$ is achieved in a distributed manner through the buses, as opposed to the copies of $K_m$ that are fairly localized and encapsulated in a nexus. 

We note that when attempting to embed $K_m\Box K_n$, it may happen that using the nexus for one of the complete graphs does not result in a valid embedding, while the other choice gives an appropriate embedding of the product. We simply choose the most promising graph and call it $K_m$ without loss of generality.


\section{Discussion}\label{sec:discussion}

In this section, we show the clear advantage of the CPCG embedding method, \emph{CPCG Embedding}, over other embedding algorithms with respect to embedding success rates and the quality of the embeddings achieved. We then prove for a specific case that CPCG Embedding is optimal with respect to the largest embeddable problem, before commenting on the scaling of the running time of the method. We begin by assuming the Chimera structure is perfectly regular (i.e., it has no inoperable qubits or couplers). The effect of irregularities is investigated in the next section. 

\subsection{Comparison to other embedding algorithms}

To showcase the advantages of our method, we compare it to the de facto heuristic method introduced by Cai, Macready, and Roy \cite{cai_2014_practical}. The implementation of this embedding algorithm, \texttt{find\_embedding()}, is distributed with D-Wave's API and software tools. This function receives both the problem to be embedded and the target solver's graph as inputs, making no assumptions about either of them. Given the NP-hardness of the embedding problem and the poor scaling of the polynomial methods when fixing the target graph,  \texttt{find\_embedding()} remains the only viable truly general alternative. The generality of the heuristic approach also has some added benefits when inoperable qubits are present, a topic that will be discussed in the next section. Since the Cartesian product of graphs $K_m$ and $K_n$ is a subset of a complete graph $K_{mn}$, the triangular systematic embedding method \cite{choi2011_minor_emb} provides a simple, yet wasteful, approach to embedding Cartesian products and will therefore serve as our second touchstone. 

Our comparison will be restricted to the specific case of embedding Cartesian products of the form $K_8 \square K_n$ into a square Chimera target architecture  $\mathcal C_{N, N, 4}$ made of bipartite blocks of 8 qubits ($K_{4,4}$). The \texttt{find\_embedding()} method is a multi-start heuristic with a number of parameters to be specified. The algorithm will keep searching until a valid embedding is found or until it reaches one of its stopping criteria. The most important parameter is the maximum running time allowed for the search, which we set to one of 1, 100, or 1000 seconds. Each restart of the search is initiated when a maximum number of steps is reached without observing an improvement. We leave this at its default value of 10 steps. We further ensure that the search is not stopped prematurely (i.e., before the maximum running time) by setting the maximum number of restarts to a large value (e.g., \mbox{10,000} restarts given that each one takes at least 1 second). 

\begin{figure}[htb!]
\begin{centering} 
\includegraphics[width=\textwidth]{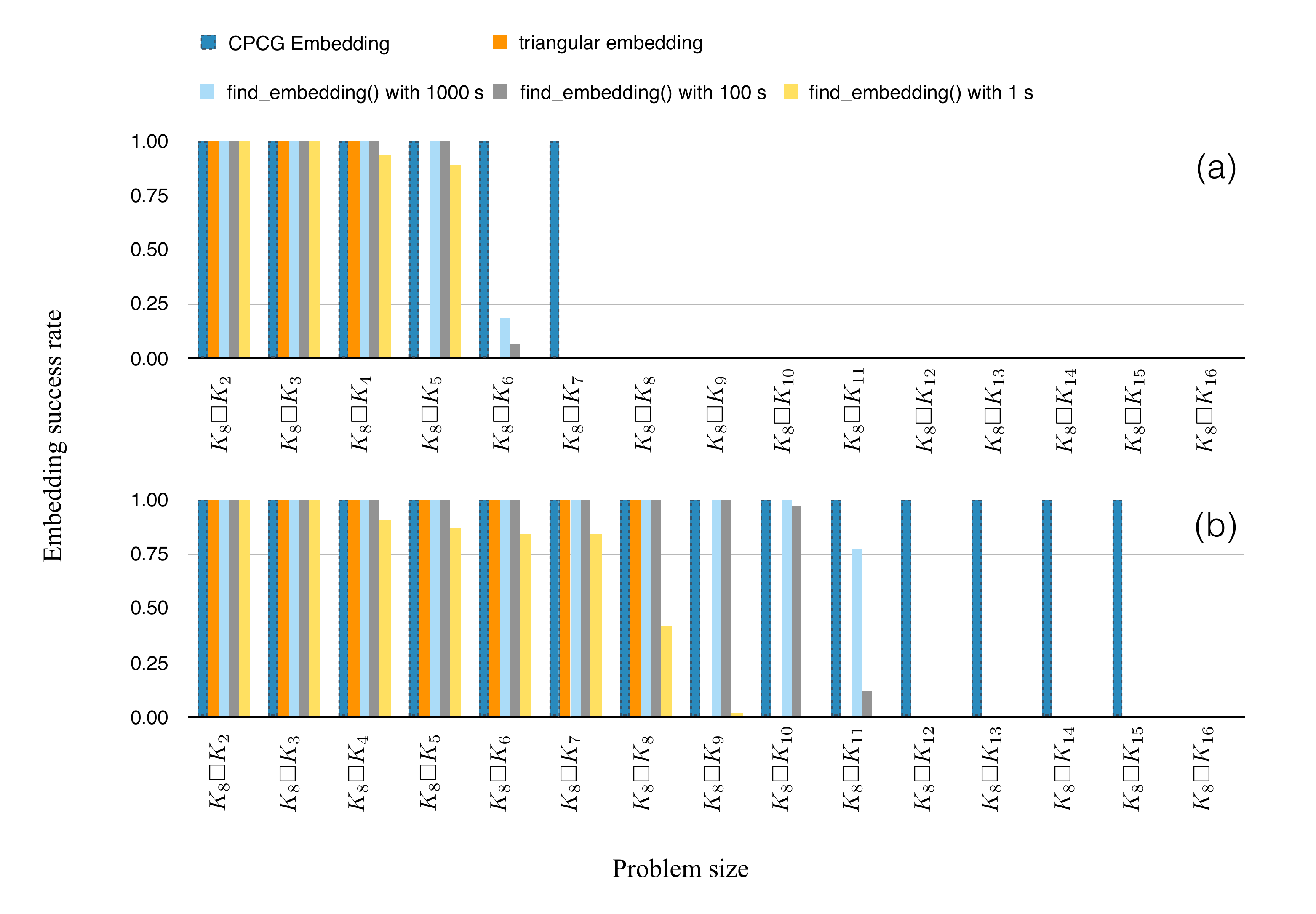}
\par\end{centering}
\caption{The embedding success rate for embedding Cartesian products of complete graphs using \texttt{find\_embedding()} (for 1000 seconds in light blue, 100 seconds in grey, and 1 second in yellow), the full triangular embedding (orange), and our systematic CPCG Embedding (dark blue) for the case of $K_8 \square K_n$ as a function of $n$. The last two are assumed to be produced in much less than a second. Panel (a) shows results for the ideal case of the previous chip's 512-qubit architecture $\mathcal{C}_{8,8,4}$, and panel (b) shows results for a hypothetical ideal 2048-qubit Chimera architecture $\mathcal{C}_{16,16,4}$.}
\label{fig:emb_figure}
\end{figure}

In our comparison, we first consider the embedding success rate of the various methods as shown in Figure \ref{fig:emb_figure}. In the case of CPCG Embedding and the triangular embedding, the success rate is simply 1.0 for all sizes smaller than some maximal size, which we can express as a function of the size $N$ of a square Chimera graph $\mathcal C_{N,N,4}$. CPCG Embedding can embed up to $K_8 \square K_{N-1}$, which means that $n=7$ is the largest case with a success rate of 1.0 for a 512-qubit chip, and $n=15$ is the largest case with that success rate for a hypothetical next-generation 2048-qubit chip. Beyond these sizes the success rate is 0. Similarly, for triangular embedding, we can embed up to $K_8 \square K_{N/2}$, which results in a success rate of 1.0 for $n \le 4$ ($n \le 8$) into a 512-qubit (2048-qubit) chip, and 0 otherwise. Since the results obtained from \texttt{find\_embedding()} are probabilistic, we attempt to embed each problem size 100 times for each maximum running time considered. For short running times, we find a satisfactory success rate only for the smallest problem sizes. We can increase that probability somewhat by increasing the running time, but even a generous 1000 seconds will not be sufficient to embed the largest Cartesian products achievable with CPCG Embedding.  Aside from the obvious effect of the poor scaling of the running time of \texttt{find\_embedding()} on the success rate, we also observe the limiting effect of the target chip's size for a fixed running time. As we get closer to the maximum embeddable problem size for a specific chip size, the success rate is further reduced. The product $K_8 \square K_6$, for example, is easily embeddable into a 2048-qubit chip, but only succeeds 18\% of the time with a 1000-second running time on the 512-qubit chip. With limited chip size also comes a limited number of valid solutions, so the probability of finding a valid solution is lower. In other words, the success rate obtained with the \texttt{find\_embedding()} method will get worse as the technology scales and we begin to address larger problem instances, but even more so when we test the limits of a specific architecture. CPCG Embedding is clearly the superior choice for a perfect Chimera chip (i.e., one where all qubits are operable), as it can embed products far larger than the two alternatives in a very short time. We discuss the scaling of running time in more detail in Section \ref{sec:runtime}. In fact, we can even show in some cases that CPCG Embedding can embed the largest possible Cartesian product of complete graphs embeddable for a target chip size (see the next section on the discussion of optimality).    

Beyond the ability to embed a problem into a chip, the quality of the embedding is paramount. Benchmarking for various types of optimization problems can show a difference of a few orders of magnitude between different embeddings of the same problem. At this point, there exists no single first-principle metric to rate embedding quality. Empirical ratings such as the metric used in \cite{Rieffel2014} represent the most-practical embedding quality metric at this point. Nevertheless, quantum annealing practitioners have used the number of physical qubits or the length of the longest chain as a conjectural measure of the embedding quality \cite{cai_2014_practical}. It has also been suggested that having heterogeneous chain lengths in an embedding is disadvantageous since the chains tend to exhibit unpredictable chain dynamics throughout the annealing process \cite{boothby2015_fast,Fully_connected_spin_glass}. Clarifying the relative role of these various properties in determining the quality of an embedding is beyond the scope of this paper, so we will limit our comparison to the traditional indicators by comparing the number of physical qubits and the chain length distribution of the CPCG Embedding with the other alternatives. 

\begin{figure}[h!]
\begin{centering}
\includegraphics[width=\textwidth]{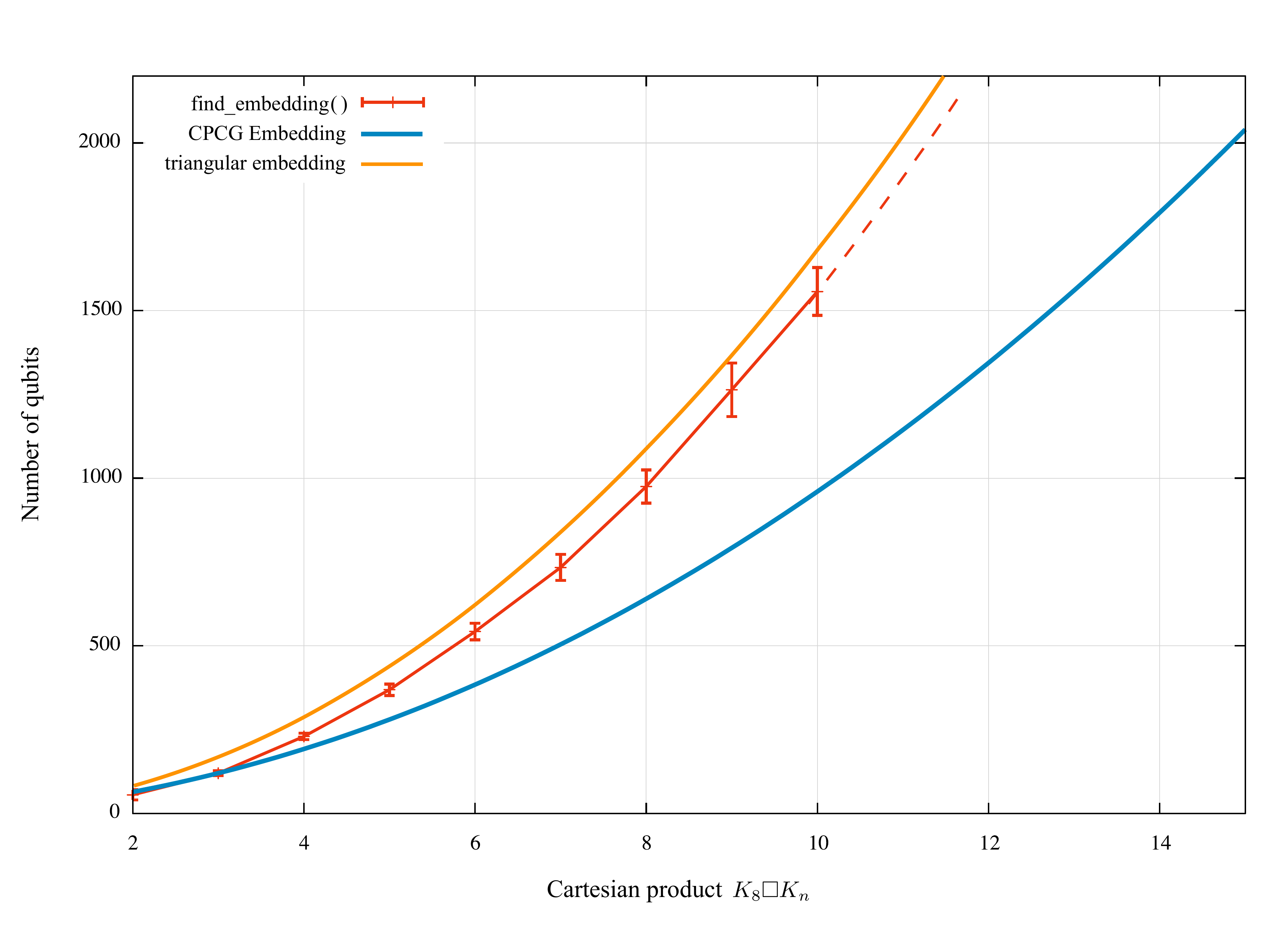}
\par\end{centering}
\caption{Average number of qubits used for embedding Cartesian products of complete graphs for the case of $K_8 \square K_n$ as a function of $n$ into a 2048-qubit architecture $\mathcal{C}_{16,16,4}$. Results are shown for D-Wave's \texttt{find\_embedding()} heuristic (in red), the full triangular embedding (orange), and our systematic CPCG Embedding (dark blue). A fit for the averaged number of qubits used in embeddings produced by \texttt{find\_embedding()} and given by $16.63n^2 -11.01n+5.75$ is shown with a dashed red line.}
\label{fig:num_qubits_figure}
\end{figure}

\begin{figure}[h!]
\begin{centering}
\includegraphics[width=\textwidth]{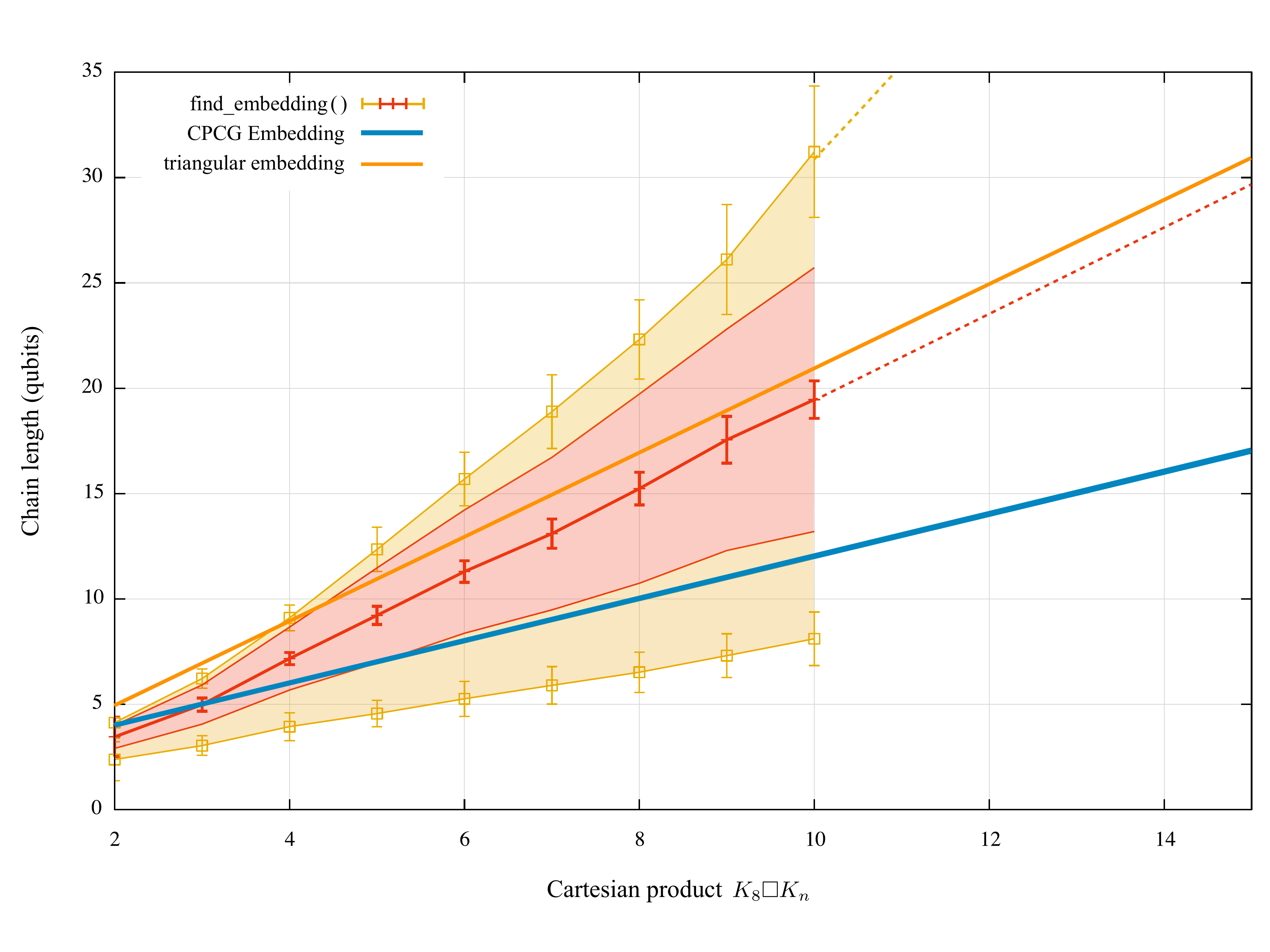}
\par\end{centering}
\caption{Chain length for embedding Cartesian products of complete graphs for the case of $K_8 \square K_n$ as a function of $n$ into a 2048-qubit architecture $\mathcal{C}_{16,16,4}$.  Results are shown for D-Wave's \texttt{find\_embedding()} heuristic (in red and yellow), the full triangular embedding (orange), and our systematic CPCG Embedding (dark blue). The latter two produce embeddings with chains that are all of equal length, shown as a single line. The spread of chain lengths produced by \texttt{find\_embedding()} is illustrated by averaging the mean (central red line), maximum (upper yellow line), and minimum (lower yellow line) chain length over 100 embeddings. The average standard deviation (also in red) of the chain length is also shown such that the red shaded region illustrates where 65\% of the chains can typically be found. A fit for the averaged maximum chain length is given by $0.12n^2+1.91n-0.48$ (dashed yellow line) and a fit for the averaged mean chain length is given by $2.05n-1.06$ (dashed red line).}
\label{fig:num_chain_figure}
\end{figure}

The number of physical qubits used is shown in Figure \ref{fig:num_qubits_figure}, and the chain length distribution is shown in Figure \ref{fig:num_chain_figure}. CPCG Embedding for an ideal Chimera graph for $K_8 \square K_n$ produces chains of length $n+2$. With $8n$ logical variables, the embedding uses a total of $8n(n+2)$ physical qubits. In comparison, a triangular embedding for a complete graph $K_{8n}$ has chains of length $2n+1$ for a total number of physical qubits used equal to $8n(2n+1)$. This is twice that of CPCG Embedding in the asymptotic limit. Both of these embedding methods produce equal-length chains. The \texttt{find\_embedding()} method, on the other hand, produces a spread of chain lengths for each successful embedding found. To illustrate this distribution, we average the mean, the minimum, and the maximum chain lengths over the 100 embeddings found. The average standard deviation is also shown such that 65\% of the chains produced are found in the red shaded region. Results depend only marginally on the maximum running time, provided that it is long enough to find a valid embedding, so we allowed for a generous 1000 seconds. A quadratic function fit of the averaged number of qubits used is given by $16.63n^2 -11.01n+5.75$, and a fit for the averaged mean chain length is given by $2.05n-1.06$. In the asymptotic limit, therefore, CPCG embedding produces chains that are less than half of the mean length produced by \texttt{find\_embedding()}. Consequently, we observe that the required number of qubits is also less than half that of the number of qubits required by \texttt{find\_embedding()}. Although \texttt{find\_embedding()} found some embeddings for $K_8 \square K_{11}$, the statistics are not shown as they were artificially skewed due to the smaller number of embeddings found. 

We find that CPCG Embedding behaves and scales favourably compared to both the heuristic method of \texttt{find\_embedding()} and the systematic triangular embedding. It can embed larger products on chips of the same size while producing shorter chains of equal length. 

\subsection{Discussion of optimality}

Having shown that CPCG Embedding compares favourably against other techniques, 
in the following theorem we  prove its optimality in certain  cases.

\begin{theorem}
Let $N$ be the smallest number such that CPCG Embedding can embed $K_m \square K_n$ 
into the square Chimera graph architecture $\mathcal{C}_{N,N,L}$.
If one of $m$ and $n$ is divisible by $2L$ and the other one is odd, then this embedding is optimal in the sense that $K_m \square K_n$  cannot be embedded into a smaller square Chimera graph.
\end{theorem}

\begin{proof}
By symmetry, we may assume that $2L$ divides  $m$ and $n$ is odd.
Let us choose $K_m$ to be the graph embedded in a nexus. 
By placing the nexuses on the diagonal, CPCG Embedding described in the previous section embeds $K_m \square K_n$ into
$\mathcal C_{N,N,L}$, where
\begin{align}
	N = \ceil[\bigg]{\frac{\ceil{\frac{m}{L}}}{2}}(n - 1) + \ceil[\bigg]{\frac{m}{L}}
	= m (n+1)  / (2L).
\end{align}

Now suppose that $K_m \square K_n$ can be embedded into $\mathcal C_{N',N',L}$.
To prove optimality we need only show that $N' \geq N$.

The proof uses a treewidth argument.
Let $\operatorname{tw}(G)$ denote the treewidth of a graph $G$.
Since $K_m \square K_n$ can be embedded into $\mathcal C_{N',N',L}$,
the former graph is a minor of the latter, so 
we have the following inequality between their treewidths:
\begin{equation}\operatorname{tw}(K_m \square K_n) \leq \operatorname{tw}(\mathcal C_{N',N',L}).\label{eq:tw}
\end{equation}

On the one hand, since $n$ is odd, we can construct a bramble similar to that given in the proof of \cite[Lemma 3.2]{lucena2007_achievable} to show that
\begin{equation}
m(n + 1)/2 - 1 \leq \textrm{tw}(K_m \square K_n). 
\end{equation}

On the other hand, the treewidth of $\mathcal C_{N',N',L}$ is known to be $N'L$
(this statement is given without proof in \cite{boothby2015_fast} and is confirmed with further explanation in \cite{private_conversation}).

Combining these two treewidth results with (\ref{eq:tw}) gives 
$$m(n + 1)/2 - 1 \leq N'L,$$
which means that
$$N' \geq m(n + 1)/(2L) - 1/L = N - 1/L.$$
But since $L> 1$ and $N$ and $N'$ are positive integers, this implies that $N'\geq N$, as required.
\end{proof}

We believe that the result above also holds for $K_m \square K_n$, where $m$ is divisible by $2L$ and $n \ge m$. However, the proof will require a complicated adaptation of the ideas given in \cite[Lemma 3.2]{lucena2007_achievable} that is beyond the scope of this paper.

As a concrete example, we consider the Chimera structure $\mathcal C_{8,8,4}$ corresponding to a 512-qubit chip. We know that CPCG Embedding can embed $K_8 \square K_7$, and, consequently, any product of the form $K_m \square K_n$, where $m \leq 8$ and $n \leq 7$. Indeed, the treewidth of $\mathcal C_{8,8,4}$ is 32, and $K_8 \square K_7$ has a treewidth smaller than or equal to 31 and is therefore embeddable, whereas any product of the form $K_8 \square K_n$ with $n > 7$ must have a treewidth of at least 35 and is not embeddable.

\subsection{Running time}\label{sec:runtime}

We may assume that the input to our algorithm is a polynomial in doubly indexed variables. For example, in the problem of colouring a graph of $n$ vertices with $m$ colours, the quadratic formulation of the problem contains a polynomial in the variables $x_{ij}$, with $i \in\{1,\dots,n\}$ and $j\in\{1,\dots,m\}$. Then, we consider the Cartesian product of two complete graphs $K_m\square K_n$ to be embedded into a Chimera graph. 

In order to identify the appropriate complete graphs whose Cartesian product contains the input graph, we need $\mathcal{O}\left(n^2m^2\right)$ operations. Furthermore, from our proposed algorithm, the total number of operations needed to embed  a Cartesian product $K_m\square K_n$ into a Chimera graph is $\mathcal O\left(n^2 m^2\right)$. It is worth mentioning that if we consider the input to be a graph with $e$ edges, the number of operations needed to identify an appropriate Cartesian product of two complete graphs is $\mathcal{O}(e)$.

Now suppose a graph $H$ is to be embedded into a graph $G$, and both of them are the inputs to the embedding algorithm proposed by \cite{cai_2014_practical}. Let $n_H$ and $e_H$ denote the number of vertices and edges of graph $H$, and $n_G$ and $e_G$ be the number of vertices and edges of graph $G$, respectively. The running time of the algorithm in \cite{cai_2014_practical} is $\mathcal{O}\left(n_Hn_Ge_H(e_G + n_G\log n_G)\right)$.

\section{Fault tolerance and extensions}

One key to our low-complexity scalable algorithm is to make use of the lattice-like regularity in the target Chimera graph. Although the numerical results show significant improvement over general heuristics used for embedding into a perfectly regular Chimera graph, we have thus far not accounted for potential defects and their impact on CPCG Embedding. One can, of course, argue that such defects are merely a temporary nuisance which will eventually be eliminated as the technology matures. Nevertheless, for the method to be of immediate practical use, the general case of a target graph with inoperable qubits and couplers needs to be considered. Unfortunately, these inoperable qubits break the perfect regularity of the Chimera graph, the very feature on which our approach is based. Figure \ref{fig:Chip509} depicts an actual instance of a D-Wave chip with inoperable qubits. This chip with a Chimera structure $\mathcal{C}_{8,8,4}$ with 509 working qubits was installed at NASA's Ames Research Center \cite{perdomo2015_DWtuning}, and was only recently replaced by a newer D-Wave 2X system. 

\begin{figure}[htbp]
\centering
\includegraphics[scale=.20]{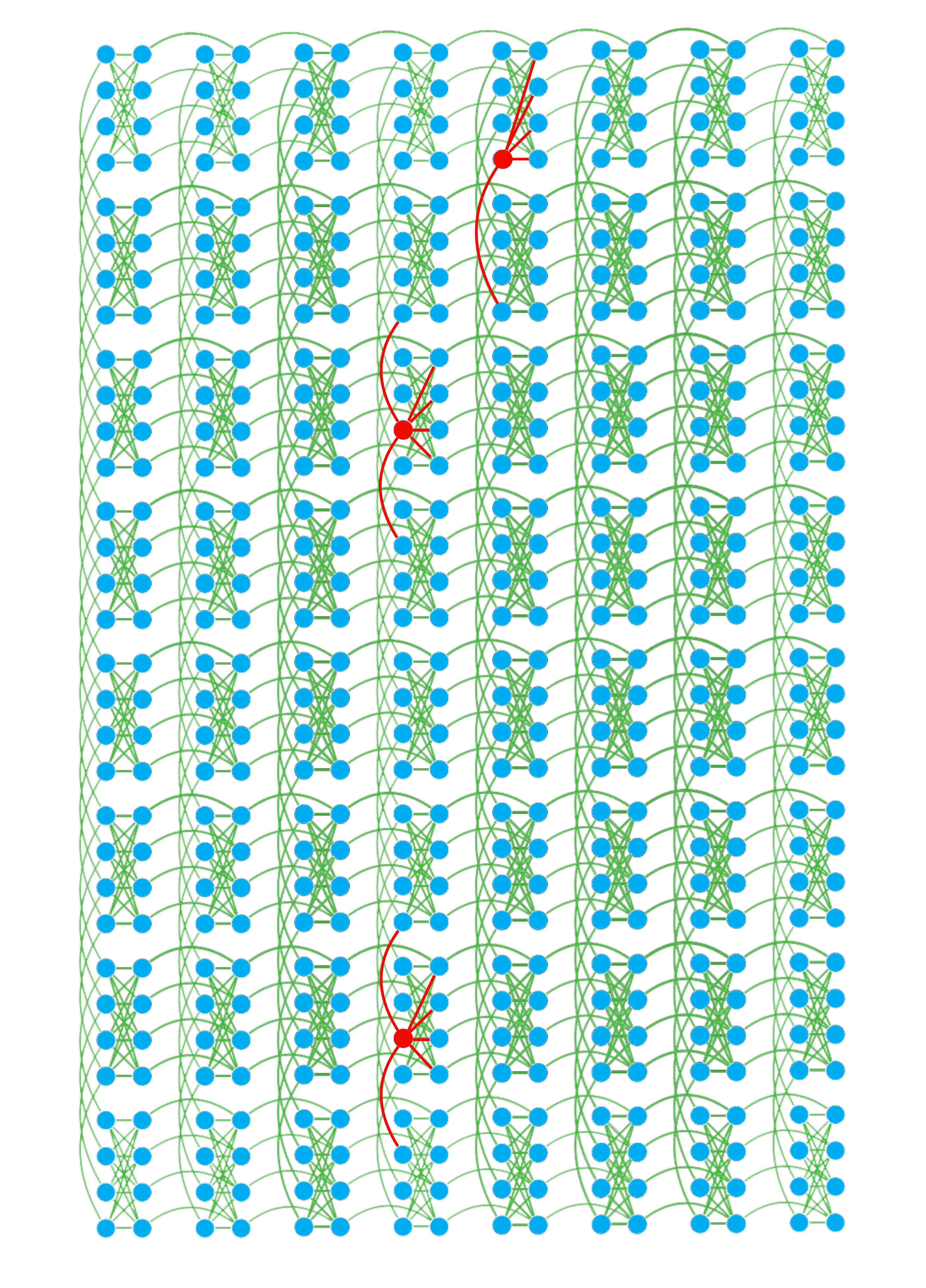}
\caption{A graphical representation of the connectivity map of the 512-qubit chip with Chimera architecture $\mathcal{C}_{8,8,4}$, and 3 inoperable qubits and associated couplers shown in red.}
\label{fig:Chip509}
\end{figure}

\subsection{Presentation of the fault-tolerant method}

The issue of having inoperable qubits can be addressed at the expense of adding more complexity to our scalable embedding approach. One simple idea is to use the CPCG embedding of the problem on an ideal solver as a starting point and apply small modifications so that a valid embedding that circumvents the irregularities caused by inoperable qubits is reached. We expect that such a solution should achieve reasonable performance on chips with high qubit yields, while a low qubit yield would lead to substantial degradation of both the embeddability and the embedding quality. We describe below how this type of extension can be implemented.

Using the embedding pattern on a perfect chip as a starting point, we address each nexus instance in turn. We begin with the first nexus and look at the capacity of its constituting blocks in each direction. The capacity of block $(i,j)$ in a given direction is denoted by $c^{i,j}_\textrm{direction}$, where $(i,j)$ are indices on a two-dimensional grid. The block capacities determine how many  paths for variables can run through a group of adjacent vertical or horizontal blocks to propagate a set of variables (i.e., the bus capacity). The presence of inoperable qubits along these directions will usually result in a reduced capacity. We then extend the size of the nexus until the relevant blocks along the vertical (horizontal) direction can form a bus with sufficient capacity. Then a variant of triangular embedding is used to embed the same complete graph in the newly extended space for the nexus. As a result of the nexus extension, we need to shift the other nexus instances appropriately. Although we have just described how a nexus extension can circumvent an inoperable qubit along a bus path, this shape modification can also help with embedding a nexus when there are inoperable qubits within the nexus boundaries. For lower qubit yields, triangular embedding might fail to embed a nexus instance regardless of the number of shifts and extensions. In such situations, a more complex nexus embedding algorithm should be used to compensate for the high irregularity in the target graph. Figure \ref{fig:nexus_emb_and_bus_config_509}a illustrates how a simple nexus extension can address the problem caused by having an inoperable qubit within the nexus, and Figure \ref{fig:nexus_emb_and_bus_config_509}b provides a more complete example by showing which modifications need to be performed to embed a $K_8 \square K_6$ on the specific \mbox{509-qubit} chip in Figure \ref{fig:Chip509}.  As the figure illustrates, the shift-and-extension method is applied to bypass the columns and rows of lower bus capacity caused by inoperable qubits. In the next section, we provide the numerical analysis of the performance of this algorithm compared to the $\texttt{find\_embedding()}$ heuristic \cite{cai_2014_practical} for this specific chip architecture. The pseudo-code in Algorithm~\ref{alg:fault_tolerant} provides a few more details of this proof of concept for this simple fault-tolerant method.

\begin{figure}
\centering     
\subfigure[A possible $K_8$ nexus extension with an inoperable qubit]{\includegraphics[scale=.38]{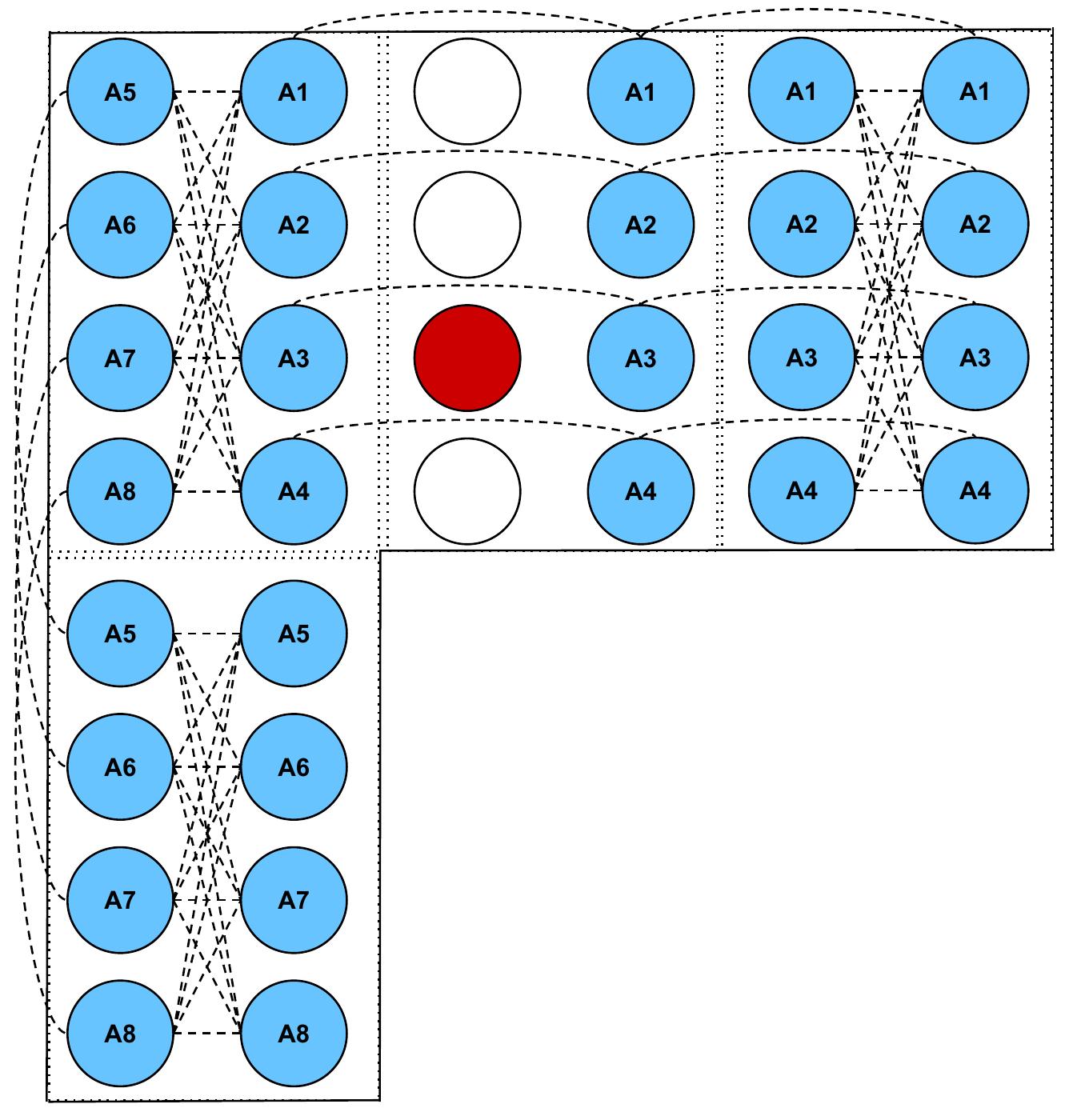} }
\mbox{ }
\subfigure[A $K_8$ nexus on a chip with inoperable qubits]{\includegraphics[scale=.5]{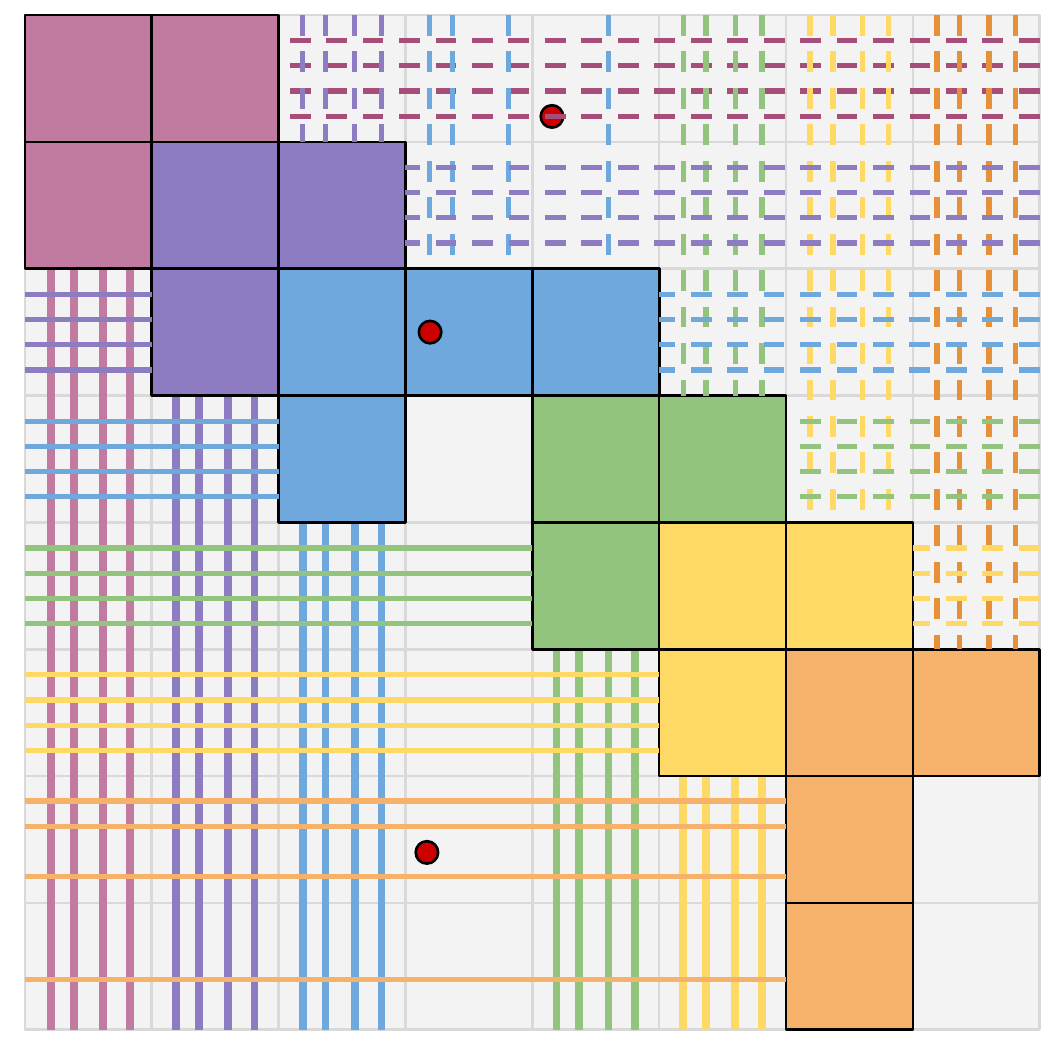}}
\caption{(a) Example of a nexus modification using a horizontal extension to avoid an inoperable qubit. The original nexus shape is the same as in Figure \ref{fig:nexus_emb_and_bus_config}a and the modification is needed to account for the inoperable qubit $A_7$, shown in red. (b) A valid embedding for the input problem $K_8 \square K_6$ into a 512-qubit chip $\mathcal{C}_{8,8,4}$ with 3 inoperable qubits, formerly installed at NASA's Ames Research Center. L-shapes with various colours are modified copies of the nexus chosen in Figure \ref{fig:nexus_emb_and_bus_config}a for embedding copies of $K_8$. The red circles represent inoperable qubits. The blue nexus is extended due to an inoperable qubit inside the nexus area and the orange nexus is extended because of a bus capacity problem caused by an inoperable qubit.}
\label{fig:nexus_emb_and_bus_config_509}
\end{figure}

\begin{algorithm}[htbp]
\scriptsize
\begin{minipage}{\textwidth}

\KwData{Adjacency matrix of the target Chimera graph $A_{\mathcal{C}_{M,N,L}}$ with inoperable qubits, dimensions of the \mbox{CPCG $(\alpha,\beta)$ }}
 \KwResult{An embedding $\mathcal{E}_{(\alpha,\beta)}$ in the target Chimera graph}
 \textbf{Initialization}: \\
 $C[i,j] \leftarrow$ \textbf{Calculate} the capacity vector $(c_{vertical}^{i,j}, c_{horizontal}^{i,j})$ for each block $[i,j]$\;
$\mathcal{E}_{(\alpha,\beta)}^{*} \leftarrow$ \textbf{Load} a scalable embedding pattern on the ideal target graph $A_{\mathcal{C}_{M,N,L}}^{*}$\;
$diagonals\leftarrow \{(Down, East), (Up, East)\}$\;
\For{direction $\in$ diagonals}{
	$(ROW,COL) \leftarrow$ coordinates of the corner block of the current $direction$\; 
	 \For{$nexus\in \mathcal{E}_{(\alpha,\beta)}^{*}$ }{
	 	  $nexus\_shape \leftarrow$ collection of blocks used by $nexus\in \mathcal{E}_{(\alpha,\beta)}^{*}$\;
		  \textbf{Shift} to the current available position $(ROW,COL)$\;
		   \While{the nexus is not embedded}{
		   $C_{nexus} \leftarrow$ required capacity by triangular embedding of $nexus$ based on $nexus\_shape$\; 
			  \eIf{$C[ROW,COL] < C_{nexus}$}{
				    $extend \leftarrow$ \textbf{identify} the direction to extend based on $sign(C_{nexus} - C[ROW,COL])$\;
				   \textbf{extend} to $(C_{nexus} - C[ROW,COL])$ blocks toward $extend$ direction\;
				   $nexus\_embeddability \leftarrow$\textbf{Call} $triangular \_ embedding()$ on updated $nexus\_shape$\;
				   \If{$nexus\_embeddability$}{
				   \textbf{Locate} $nexus$ on $A_{C_{M,N,L}}$\;
				   \textbf{Update} $\mathcal{E}_{(\alpha,\beta)}$ and $(ROW,COL)$\;
				   \textbf{continue} to next $nexus$\;
				   }
				   }{
				   \textbf{Locate} $nexus$ on $A_{\mathcal{C}_{M,N,L}}$\;
				   \textbf{Update} $\mathcal{E}_{(\alpha,\beta)}$ and $(ROW,COL)$\;
			  }
		  }
	}
 }
\end{minipage}
\caption{Fault-tolerant CPCG Embedding based on shifts and extensions}
\label{alg:fault_tolerant}
\end{algorithm}

\subsection{Comparison to other embedding methods}

We have tested the simple fault-tolerant algorithm to embed the family of $K_8 \square K_n$ problems on the quantum annealer described in Figure \ref{fig:Chip509}. We again compare it to the results produced by the \texttt{find\_embedding()} heuristic run for 1000 seconds and each problem was repeated 100 times to collect statistics. The other parameters used are the same as in Section \ref{sec:discussion}. Figure \ref{fig:nexus_emb_and_bus_config_509}  shows an embedding of the maximum problem size embeddable on this chip with CPCG Embedding. The \mbox{\texttt{find\_embedding()}} heuristic can also embed this problem size, albeit with a success rate of about 18\%. Unsurprisingly, the success rate of a heuristic method such as \texttt{find\_embedding()} is not greatly affected by a small number of inoperable qubits. Figure \ref{fig:emb_figure_509} illustrates how the success probability of \texttt{find\_embedding()} still drops faster than CPCG Embedding with increasing problem size. We expect our previous observation that the advantage of CPCG Embedding becomes more prominent for larger chip sizes to hold for high qubit yields. In other words, CPCG Embedding should outperform \texttt{find\_embedding()} by a larger margin for larger target architectures despite the presence of irregularities caused by a low density of inoperable qubits. 

We also note that the embedding quality of CPCG embeddings remains superior with respect to the number of required qubits and chain length distribution. These are compared in Figures \ref{fig:num_qubits_509} and \ref{fig:chain_509} for the chip with 509 qubits. Here, too, the results are not shown for $K_8\square K_6$ for \texttt{find\_embedding()} because they were skewed due to a lower embedding success rate. Again, we observe that the \texttt{find\_embedding()} heuristic is not very sensitive to this small density of inoperable qubits, as the chain length distribution and qubit counts are almost identical to the ideal case. Given that CPCG Embedding relies on the regularity of the target graph, unlike  \texttt{find\_embedding()}, we unsurprisingly observe a higher sensitivity to the irregularities caused by inoperable qubits. This results in some degradation in the embedding quality. The chains in each successful embedding are no longer equal because the algorithm needs to route around inoperable qubits, resulting in the spreading out of the distribution. For the same reason, we observe a larger qubit count for the embeddings on the real chip compared to the ideal case. Despite the changes, both the required number of qubits and the distribution of chain lengths remain significantly superior to \texttt{find\_embedding()}. It is true that we are considering a high qubit yield, but Figures \ref{fig:num_qubits_509} and \ref{fig:chain_509} indicate that CPCG Embedding still has potentially enough of an advantage over \texttt{find\_embedding()} to remain the preferable method, even for lower qubit yields. Obviously, a crossing point is expected and methods like \texttt{find\_embedding()} remain indicated for irregular target graphs. 

We included this simple algorithm to show the possibility of modifying  our approach to be used for real chips. However, this approach seems intuitively wasteful as it readily discards large blocks of qubits. There exists an obvious trade-off between the complexity of the fault-tolerant embedding algorithm and its performance in terms of embedding success rate and embedding quality. Work on a refined approach, still based on modifying the ideal CPCG embedding pattern, is ongoing and will be presented elsewhere. We believe that improvements to the techniques described herein should allow us to achieve a higher tolerance to irregularities while preserving most of the desirable features such as running time and embedding properties.    

\begin{figure}[htbp!]
\begin{centering}
\includegraphics[scale=0.52]{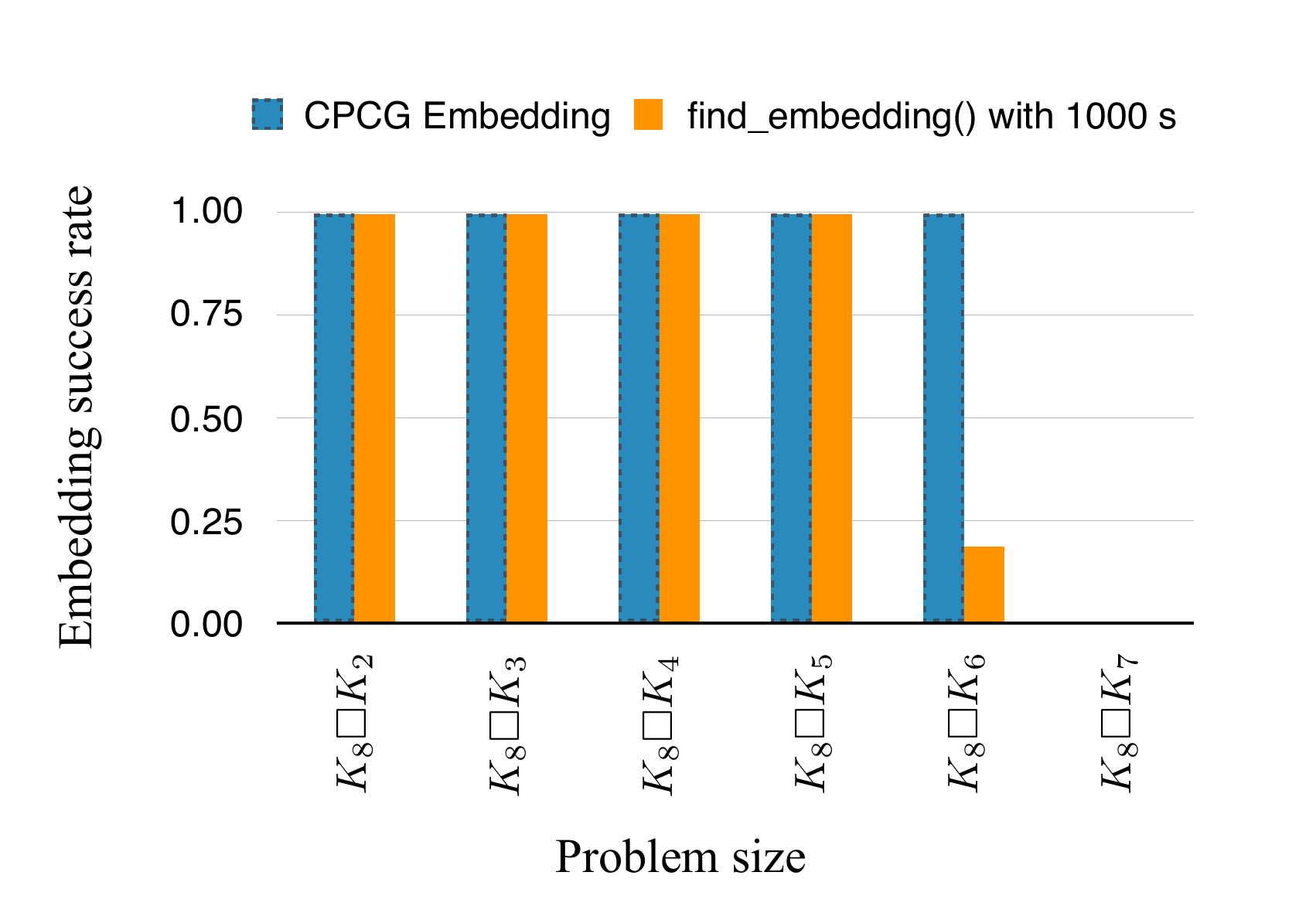}
\par\end{centering}
\caption{The embedding success rate for embedding Cartesian products of complete graphs into a chip with inoperable qubits using D-Wave's \texttt{find\_embedding()} heuristic for 1000 seconds (orange) and our systematic CPCG Embedding (dark blue) for the case of $K_8 \square K_n$ as a function of $n$. The $\mathcal{C}_{8,8,4}$ chip used has 509 working qubits out of 512 and is described in Figure \ref{fig:Chip509}.  The largest problem embedded by both approaches is $K_8 \square K_6$ with a success rate of 19\% for \texttt{find\_embedding()} and 100\% for CPCG Embedding. }
\label{fig:emb_figure_509}
\end{figure}

\begin{figure}[hbtp!]
\begin{centering}
\includegraphics[width=\textwidth]{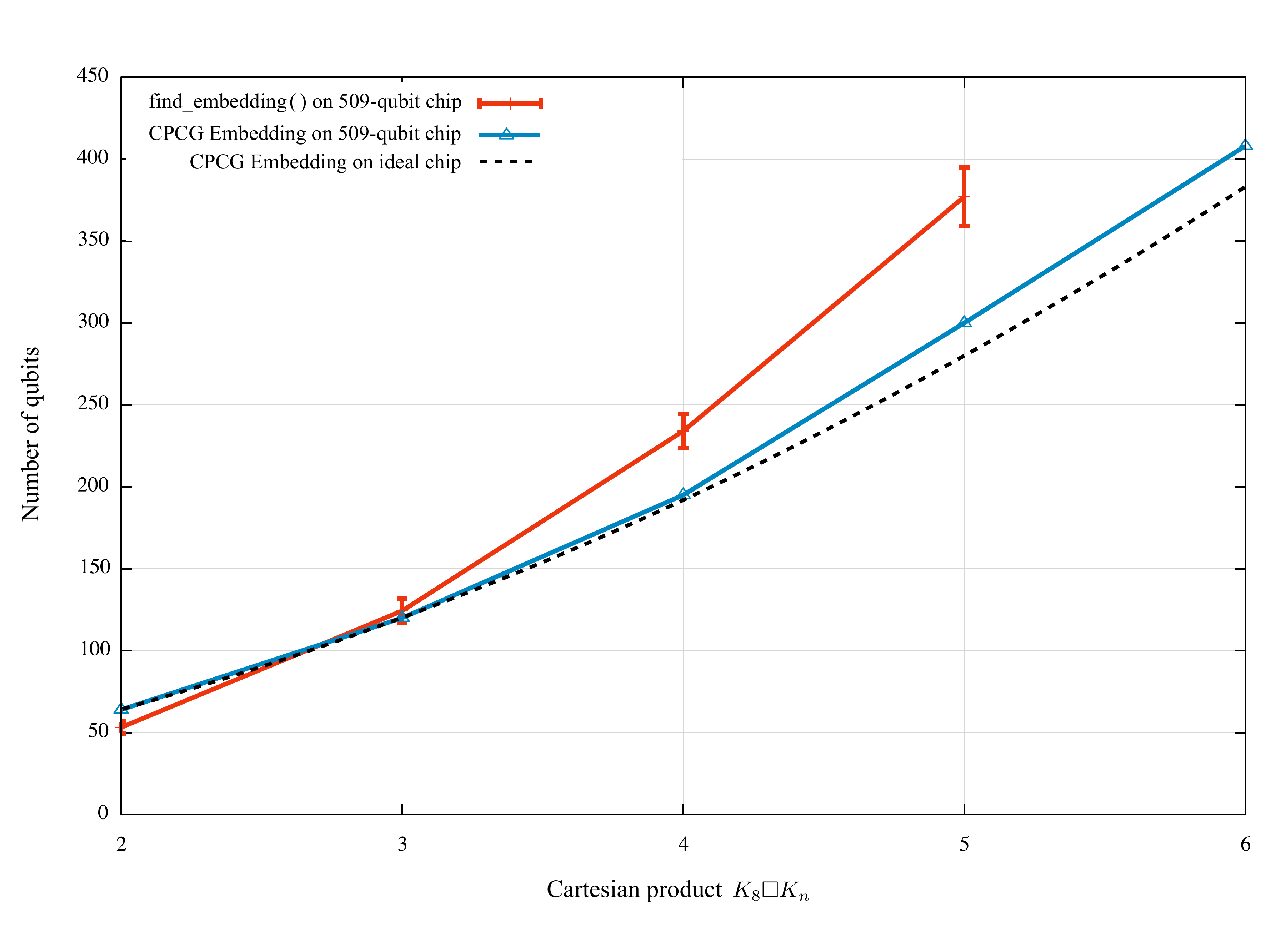}
\par\end{centering}
\caption{Average number of qubits used for embedding Cartesian products of complete graphs into a chip with inoperable qubits using D-Wave's \texttt{find\_embedding()} heuristic for 1000 seconds  (red) and our systematic CPCG Embedding (dark blue) for the case of $K_8 \square K_n$ as a function of $n$. The $\mathcal{C}_{8,8,4}$ chip used has 509 working qubits out of 512 and is described in Figure \ref{fig:Chip509}. CPCG Embedding results for a perfect chip of the same size are also shown (dotted black line). }
\label{fig:num_qubits_509}
\end{figure}

\begin{figure}[hbtp!]
\begin{centering}
\includegraphics[width=\textwidth]{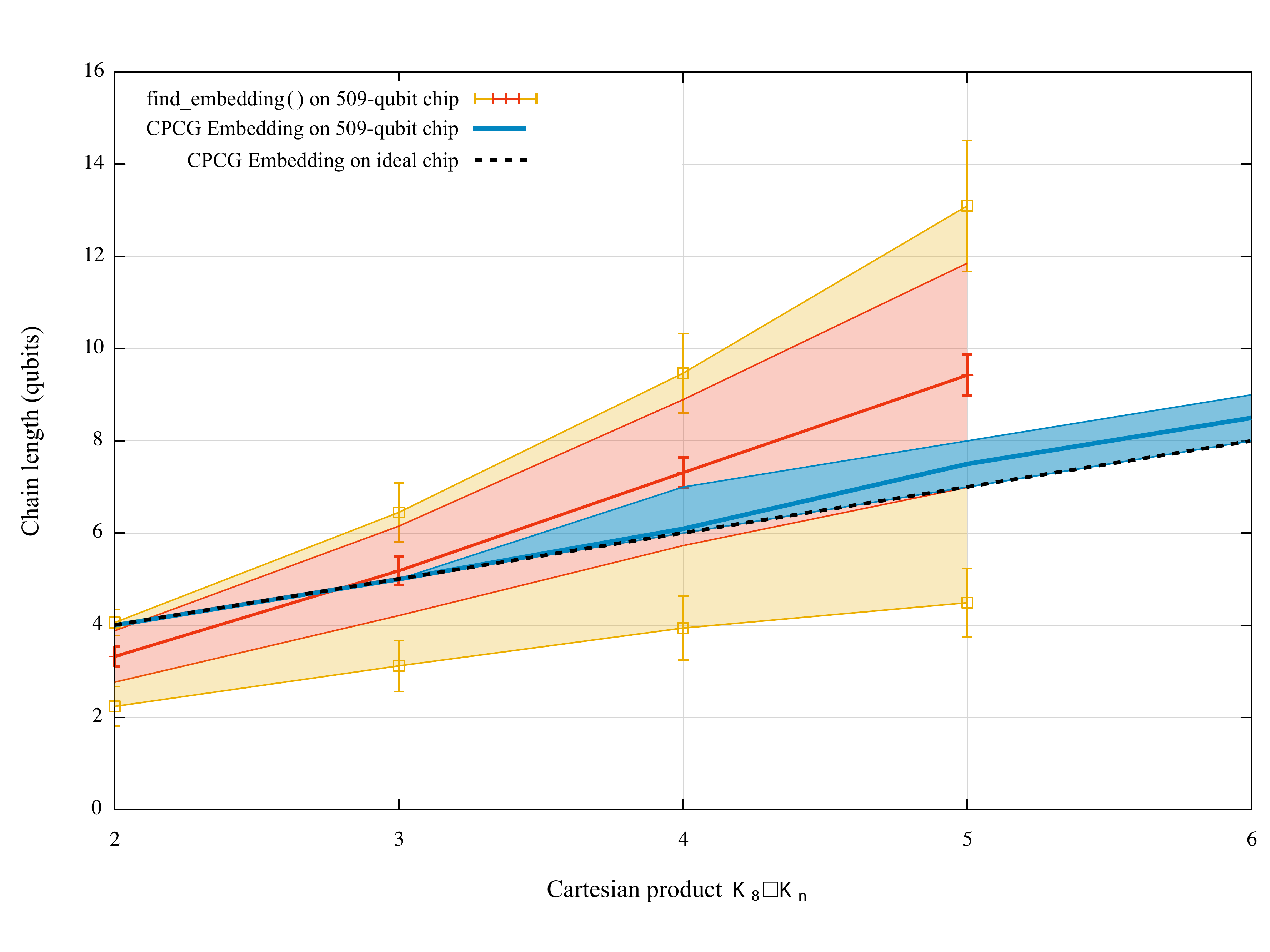}
\par\end{centering}
\caption{Chain length for embedding Cartesian products of complete graphs into a chip with inoperable qubits using D-Wave's \texttt{find\_embedding()} heuristic for 1000 seconds (yellow and red) and our systematic CPCG Embedding (blue) for the case of $K_8 \square K_n$ as a function of $n$. The $\mathcal{C}_{8,8,4}$ chip used has 509 working qubits out of 512 and is described in Figure \ref{fig:Chip509}.  One CPCG embedding instance for each problem size is shown in blue with the dark blue line showing the average chain length and the shaded blue area representing the spread between the minimum and maximum chain lengths. The spread of chain lengths produced by \texttt{find\_embedding()} is illustrated by averaging the mean (central red line), maximum (upper yellow line), and minimum (lower yellow line) chain lengths over 100 embeddings. The average standard deviation (also in red) of the chain length is also shown such that the red shaded region illustrates where 65\% of the chains can typically be found.}
\label{fig:chain_509}
\end{figure}

\section{Conclusion}

Motivated by several interesting combinatorial problems such as graph colouring and graph partitioning, we proposed a systematic, deterministic, and scalable embedding algorithm for embedding the Cartesian product of two complete graphs into D-Wave Systems' Chimera hardware graph. To develop this method, we exploited the intrinsic structure of a class of combinatorial optimization problems as well as the structure of the Chimera graph. Although more-general (and perforce slower) methods will remain necessary, it is with such application-specific algorithms that the best performance can be achieved and we expect similar studies to follow suit in the near future. In the case of our CPCG Embedding algorithm, we not only showed advantageous running time scaling, and how embedding patterns can be cost-effectively scaled up for larger chip architectures, we also proved CPCG Embedding to be optimal in specific cases. Beyond the better embedding success rate achieved, the quality of the embeddings generated, as measured by the usual empirical factors, is superior to other methods. Indeed, CPCG Embedding produces equal-length chains on an ideal Chimera chip and uses far fewer physical qubits. Such improvements in the quality of embedding can play a major role in reducing the time to solution when solving problems. Given the algorithm's reliance on the regularity of the target architecture, it is natural to expect a degradation of performance in the presence of defects. Although we did not explore optimal modifications to the method to handle  inoperable qubits and couplers, we presented a simple version of the algorithm for those cases and tested it on the $\mathcal{C}_{8,8,4}$ 512-qubit NASA chip with 509 working qubits, described in \cite{perdomo2015_DWtuning}. The results suggest that for high qubit yields, CPCG Embedding will retain some advantage in both embedding success rates and quality indicators over more-general heuristic methods.

\begin{acknowledgements}
The authors are grateful to Marko Bucyk for editing the manuscript, and to Brad Woods, Natalie Mullin, Abbas Mehrabian, and Robyn Foerster for useful discussions and input. 
\end{acknowledgements}

\bibliographystyle{spmpsci}     
\bibliography{CPCG}

\end{document}